\documentclass{birkjour}
\usepackage{amsmath,tikz,amssymb,color,array}

\newtheoremstyle{indenteddefinition}{8pt}{8pt}{\addtolength{\leftskip}{2.5em}\addtolength{\rightskip}{2.5em}}{-2.5em}{\bfseries}{.}{.5em}{{\thmname{#1 }}{\thmnumber{#2}}{\thmnote{ (#3)}}}
\newtheorem{thm}{Theorem}[section]
\newtheorem{cor}[thm]{Corollary}
\newtheorem{lem}[thm]{Lemma}
\newtheorem{prop}[thm]{Proposition}
\theoremstyle{indenteddefinition}
\newtheorem{defn}[thm]{Definition}
\theoremstyle{indenteddefinition}

\numberwithin{equation}{section}
\newcommand{\C}{\mathfrak{C}}

\renewcommand{\P}{\mathbb{P}}
\renewcommand{\Pr}{\mathcal{P}}

\newcommand{\Z}{\mathbb{Z}}
\renewcommand{\O}{\mathcal{O}}
\newcommand{\fix}{\mbox{fix}}
\newcommand{\gn}{\mbox{gn}}
\newcommand{\Fix}{\mbox{Fix}}

\renewcommand{\b}{\beta}

\title[Guessing Numbers of Odd Cycles]{Guessing Numbers of Odd Cycles}

\author[Atkins]{Ross Atkins}
\address{
University of Oxford \\
Department of Statistics \\
1 South Parks Road \\
Oxford OX1 3TG \\
United Kingdom}
\email{ross.atkins@univ.ox.ac.uk}

\author[Rombach]{Puck Rombach}
\address{
University of California, Los Angeles \\ 
Department of Mathematics \\
520 Portola Plaza \\
CA 90095-1555\\
United States 
}
\email{rombach@math.ucla.edu}

\author[Skerman]{Fiona Skerman}
\address{
Heilbronn Institute \\
University of Bristol \\
Howard House, Queen's Ave\\
Bristol BS8 1SD\\
United Kingdom}
\email{f.skerman@bristol.ac.uk}

\begin{document}

\begin{abstract}
For a given number of colours, $s$, the guessing number of a graph is the base $s$ logarithm of the size of the largest family of colourings of the vertex set of the graph such that the colour of each vertex can be determined from the colours of the vertices in its neighbourhood. An upper bound for the guessing number of the $n$-vertex cycle graph $C_n$ is $n/2$. It is known that the guessing number equals $n/2$ whenever $n$ is even or $s$ is a perfect square \cite{Christofides2011guessing}. We show that, for any given integer $s\geq 2$, if $a$ is the largest factor of $s$ less than or equal to $\sqrt{s}$, for sufficiently large odd $n$, the guessing number of $C_n$ with $s$ colours is $(n-1)/2 + \log_s(a)$. This answers a question posed by Christofides and Markstr\"{o}m in 2011~\cite{Christofides2011guessing}. We also present an explicit protocol which achieves this bound for every $n$. \\ Linking this to index coding with side information, we deduce that the information defect of $C_n$ with $s$ colours is $(n+1)/2 - \log_s(a)$ for sufficiently large odd $n$. Our results are a generalisation of the $s=2$ case which was proven in~\cite{bar2011index}.
\end{abstract}

\maketitle


\section{Introduction}

Computing the guessing number (Definition~\ref{defn:gn}) of a graph $G$, can be equivalent to determining whether the multiple unicast coding problem \cite{Dougherty2006nonreversibility} is solvable on a network related to $G$. The guessing number of a graph, $G$, is also studied for its relation to the information defect of $G$ and index coding with side information \cite{Alon2008broadcasting,Gadouleau2011graph}. Exact guessing numbers are known only for a small number specific classes of graphs, such as perfect graphs, or small cases of non-perfect graphs  \cite{Baber2014graph,Cameron2014guessing,Chang2014linear,Wu2009guessing}.   In particular, the guessing number of odd cycles, which is the focus of this paper, was not known, except for small cases \cite{Christofides2011guessing,bar2011index}. Here we compute the guessing number of the cycle graph, $C_n$, by analysing optimal protocols for the ``guessing game''. 

The guessing game was introduced by Riis in 2007 \cite{Riis2007information}. It is a cooperative $n$-player information game played on a graph with $n$ vertices with $s$ colours. The guessing game on the complete graph $K_n$ with $s=2$ colours is played as follows. Each of the $n$ players are given a hat that is red or blue uniformly and independently at random. Each player can see everyone else's hat, but not their own. The players collaboratively aim to maximise the probability that all players guess the colour of their hats correctly. Much of the popularity of this puzzle is owed to the striking difference between the success probability achieved by uncoordinated random guessing and an optimal protocol, which are $1/2^n$ and $1/2$ respectively. 

The general guessing game considered here differs from many other variants of multiplayer information games (for example: the ``hat guessing game'' \cite{Butler2009hat}, ``Ebert's game'' \cite{Ebert1998applications} and the ``hats-on-a-line game" \cite{Patterson2010another}) in the following critical ways:
	\begin{itemize}
	\item The colours are assigned to each player independently and uniformly.
	\item Every player must guess (no passing or remaining silent).
	\item Each player does not necessarily see every other player's colours; two players can see each other if and only if they are joined by an edge in a given graph $G$.
	\item The players guess simultaneously so no communication is possible once the colours are assigned. 
	\item The guessing game is won only if \emph{all} the players guess correctly. An incorrect guess by any single player would mean that the whole team of $n$ players collectively lose the guessing game (unlike \cite{Butler2009hat}, for example which seeks to optimise the number of players who guess correctly).
	\end{itemize}
It is known that the greatest probability of winning the guessing game can be achieved by a deterministic protocol \cite{Cameron2014guessing}. Let $G = (V,E)$ be a graph where $V = \{ v_1,v_2, \ldots ,v_n \}$ is the set of vertices and $E \subseteq \binom{V}{2}$ is the edge set. We restrict our attention to undirected graphs, but the problem is generalizes to directed graphs in an obvious way.
\begin{defn}[Protocol, colouring]
\label{defn:prot}
For any positive integer $s$, we let $\Z_s$, the group of all residues modulo $s$, denote the \emph{colour set}. 
A \emph{colouring} of $G$ with $s$ colours is an $n$-tuple $c = (c_1,c_2, \ldots ,c_n)$ such that $c_i \in \Z_s$. 
The set of all colourings of $G$ with $s$ colours is denoted $\Z_s^n$. 
A \emph{protocol} on  $G$ with $s$ colours is any $n$-tuple $\Pr = (f_1,f_2,f_3, \ldots ,f_n)$ where for each $i$, the [deterministic] function $f_i : \Z_s^n \rightarrow \Z_s$ is such that $f_i(c)$ is dependent only on $c_j$ for all $j$ such that $v_iv_j \in E$, {\it i.e.}
for any $i$ and any two colourings $c = (c_1,c_2, \ldots ,c_n)$ and $c^\prime = (c_1^\prime,c_2^\prime,\ldots ,c_n^\prime)$, if $c_j^\prime = c_j$ for all $j$ such that $v_iv_j \in E$ then $f_i(c) = f_i(c^\prime)$. The \emph{fixed set} of $\Pr$, $\Fix(\Pr)$, is the set of all invariant colourings:
	$$ \Fix(\Pr) = \big\{ c \in \Z_s^n \: | \: c_i=f_i(c) \; \forall i \big\}. $$
	\end{defn}
\begin{defn}[Fixed number, fixed set]
\label{defn:gn} 
The \emph{fixed number} of $\Pr$ is the size of its fixed set; $\fix(\Pr) = |\Fix(\Pr)|$. 
A protocol $\Pr$ is called \emph{non-trivial} if $\Fix(\Pr) \not= \emptyset$. 
A protocol is called \emph{optimal} if it has maximal fixed number. 
\end{defn}
\begin{defn}[Guessing number]
\label{defn:fixed} 
The \emph{guessing number} of $G$ with $s$ colours is defined as
	$$ \gn(G,s) = \log_s \max_\Pr \left[ \fix(\Pr) \right]. $$
\end{defn}

We assign the $n$-tuple of colours $c \in \Z_s^n$ uniformly at random to the set of players, who are each identified with a vertex of $G$.
The guesses of the players are given by $\Pr(c)$, so the players win if and only if $c = \Pr(c)$. Hence, the probability that an optimal protocol $\Pr$ wins is
	$$ \P\big(c = \Pr(c)\big) = \frac{\fix(\Pr)}{|\Z_s^n|} = s^{\gn(G,s)-n}. $$
Christofides and Markstr\"{o}m \cite{Christofides2011guessing} showed that, for a perfect graph $G$ and any $s$, $\gn(G,s) = n-\alpha$ where $\alpha$ is the size of the largest independent set in $G$. For example, the complete graph $K_n$ is a perfect graph with $\alpha = 1$, so an optimal protocol on $K_n$, wins with probability $1/s$. The $3$-cycle and the even-cycle $C_{2k}$ (for any positive integer $k$) are both perfect graphs with $\alpha(C_3)=1$ and $\alpha(C_{2k})=k$ so 
	\begin{equation}
	\label{eq:perfect_graphs}
	\gn(C_3,s) = 2 \quad \mbox{and} \quad \gn(C_{2k},s)=k \quad \forall \: k. 
	\end{equation}
Henceforth, we shall consider only the cycle graphs $C_n$ for odd $n \geq 5$. In \cite{Christofides2011guessing}, it is shown that $$\gn(C_5,2) =5,$$
and the analysis in \cite{bar2011index} shows that
$$\gn(C_n,2) = \frac{n-1}{2},\;\; \mbox{for odd }n\geq 7. $$
For general $s$, Christofides and Markstr\"{o}m define protocols called ``the clique strategy'' and ``the fractional-clique strategy'' \cite{Christofides2011guessing}. The fractional clique strategy is only defined when the number of colours $s$ is a perfect power, and it is shown to be optimal on the odd cycle whenever $s$ is a perfect square, {\it i.e.} 
	\begin{equation}
	\label{eq:s_square}
	\gn(C_n,m^2) = \frac{n}{2} \qquad \forall \: n,m. 
	\end{equation}
In Definition~\ref{defn:P_fcp}, a protocol $\Pr_{fcp}$ is defined on odd cycles for any number of colours $s$. The protocol $\Pr_{fcp}$ is equivalent to the clique-strategy when $s$ is prime, and to the fractional-clique-strategy when $s$ is a perfect square. The protocol $\Pr_{fcp}$ is called the \emph{fractional-clique-partition protocol} to emphasise that it is very closely related to Christofides and and Markstr\"{o}m's fractional-clique strategy. Our main result in Theorem~\ref{thm:n_large} states that, for any given $s$, this fractional-clique-partition protocol is optimal on any large enough odd cycle.

The rest of this paper is organised as follows. In Section \ref{sec:notation}, we summarise a few of the known results on guessing numbers, and introduce the concepts of entropy and mutual information, which we will use heavily in our proofs. In Section \ref{sec:frac}, we define the fractional-clique-partition protocol, which is a refinement of the protocol introduced in \cite{Christofides2011guessing} and we prove that for odd $n$, as the number of colours grows, this protocol achieves a $\fix(\Pr)$ lies between $s^{n/2}$ and $s^{n/2}(1-\O(n/\sqrt{s}))$ (Theorem \ref{thm:round_up}). In Section \ref{sec:entropy}, we lay the technical groundwork which is needed for Section~\ref{sec:n_large}. Then, in Section \ref{sec:n_large}, we focus on the case of large $n$ compared to $s$, and we prove that the fractional-clique-partition protocol is in fact optimal on large enough odd cycles (Theorem~\ref{thm:n_large}). In Section~\ref{sec:defective}, we link this to index coding with side information and compute the size of an optimal index code for $C_n$ with $s$ colours when $n$ is odd and sufficiently large.

\section{Backround Material and Notation}
\label{sec:notation}

Many of our proofs will use the concept of the entropy of a random variable. Entropy is defined in Definition~\ref{defn:entropy} and we list three crucial properties in Proposition~\ref{prop:known_entropy_results}. 
In this paper we take most logarithms base $s$, including inside the definitions of entropy.
In the rest of this section, we present a few known results on the guessing number, define some useful random variables on the cycle graph and a notion of entropy, all of which will be used extensively in our proofs. When possible, we are consistent with the definitions and notations given in \cite{Cameron2014guessing,Chang2014linear,Christofides2011guessing,Baber2014graph,Riis2007graph,Riis2007information}. We start with a small, useful result that shows, intuitively, that we are allowed to ``forget'' some colours. 

\begin{prop}
\label{prop:monotone}
Let $G$ be a graph, let $s$ and $s^\prime$ be positive integers with $s^\prime \leq s$, and let $\Pr$ be any protocol on $G$ with $s$ colours. There exists a protocol $\Pr^\prime$ on $G$ with $s^\prime$ colours such that 
	$$ \big\{ c \in \Fix(\Pr) \: \big| \: 0 \leq c_i < s^\prime \: \forall i \big\} \subseteq \Fix(\Pr^\prime). $$
\end{prop}
\begin{proof}
If $\Pr = (f_1,f_2, \ldots f_n)$ then define $\Pr^\prime = (f_1^\prime,f_2^\prime, \ldots ,f_n^\prime )$ in the following way.
\begin{itemize}
\item If $0 \leq c_j<s^\prime$ for all $j$ such that $v_iv_j \in E$, and $0 \leq f_i(c)<s^\prime$ then $f_i^\prime(c) = f_i(c)$.
\item If $s^\prime \leq c_j < s$ for any $j$ such that $v_iv_j \in E$, or $s^\prime \leq f_i(c) < s$ then $f_i^\prime(c) = 0$.
\end{itemize}
For any colouring $c \in \Fix(\Pr)$, if $0 \leq c_i < s^\prime$ for all $i$, then $\Pr^\prime(c) = \Pr(c) = c$ so $c \in \Fix(\Pr^\prime)$.
\end{proof}

\begin{defn}[Entropy, mutual information]
\label{defn:entropy}
Let $A_1, \ldots, A_k$ be random variables which take values in a finite set $\mathcal{A}$. The \emph{entropy} of $A_1,\ldots, A_k$ is denoted $H(A_1,\ldots, A_k)$ and is given by:
	$$ H(A_1,\ldots,A_k) = -\sum_{a_1,\ldots,a_k \in \mathcal{A}^k} \P(A_1=a_1,\ldots,A_k=a_k) \log_s \P(A_1=a_1,\ldots,A_k=a_k). $$	
The \emph{mutual information} of $A_1$ and $A_2$ is denoted $I(A_1;A_2)$ and is given by: 
	$$ I(A_1;A_2) = H(A_1) + H(A_2) - H(A_1,A_2). $$
	Let $B$ be another random variable taking values in $\mathcal{A}$. The conditional mutual information of $I(A_1;A_2|B)$ is given by
	\begin{equation}
	I(A_1;A_2|B) = H(A_1,B) + H(A_2,B) - H(A_1,A_2,B) - H(B).
	\label{eq:conditional_mutual_information}
	\end{equation}
\end{defn}

\begin{prop}
\label{prop:known_entropy_results}
Let $A_1$, $A_2$ be random variables which take values in a finite sets $\mathcal{A}$. 
\begin{enumerate}
\item $H(A_1) \leq \log |\mathcal{A}|$ with equality if and only if $A_1$ is uniformly distributed. 
\item $I(A_1;A_2) \geq 0$ with equality if and only if $A_1$ and $A_2$ are independent.
\item $I(A_1;A_2|B) \geq 0$ with equality if and only if $A_1$ and $A_2$ are independent conditional on $B$. 
\end{enumerate}
\end{prop}

For a proof of the results in Proposition~\ref{prop:known_entropy_results} we refer the reader to \cite{Cover2006elements}. 

\begin{defn}
For a non-empty set $S$, we use the notation $A \in_u S$ to mean $A$ is a random variable distributed uniformly over all elements in $S$. 
\end{defn}

\begin{defn}[Notation for $\mathbf{C_n}$]
\label{defn:C_n}
The cycle graph, $C_n$, has $n$ vertices $V = \{ v_1,v_2, \ldots ,v_n \}$. The edge set of $C_n$ is 
	$$ E = \{ v_iv_{i+1} \: | \: i = 1,2,3, \ldots ,n \} $$
(indices are always taken modulo $n$). In a slight abuse of notation, for any protocol $\Pr = (f_1,f_2,f_3, \ldots ,f_n)$ on $C_n$ with $s$ colours, we say $f_i : \Z_s^2 \rightarrow \Z_s$ where 
	$$ f_i(c) = f_i(c_{i-1},c_{i+1}). $$
Recall that a protocol $\Pr$ is non-trivial if $\Fix(\Pr) \neq \emptyset$. For a given non-trivial protocol $\Pr$ on $C_n$, define $X = (X_1,X_2, \ldots ,X_n)$ to be a colouring chosen uniformly at random from $\Fix(\Pr)$. \emph{i.e.} 
	$$ X \in_u \Fix(\Pr). $$
Note that the random colouring $X = (X_1,X_2, \ldots ,X_n)$ is only defined for non-trivial protocols $\Pr$. 
To simplify notation we will sometimes denote the entropy of a tuple of $X_i$s by 
	$$ h(i_1,i_2,i_3, \ldots ) = H(X_{i_1},X_{i_2},X_{i_3}, \ldots ). $$
Since $X_i$ is determined by $(X_{i-1},X_{i+1})$ we must have $H(X_{i-1},X_i,X_{i+1}) = H(X_{i-1},X_{i+1})$ so $h(i-1,i,i+1)=h(i-1,i+1)$. In general we can freely remove the argument $i$ from $h( \ldots , i-1 , i , i+1 , \ldots )$ as long as we don't remove the arguments $i-1$ and $i+1$. 
	\begin{equation}
	\label{eq:awesome}
	h( \ldots , i-1 , i , i+1 , \ldots ) = h( \ldots , i-1 , i+1 , \ldots )
	\end{equation}
To simplify notation even further, for integers $j<k$, let $H_j^k$ denote the quantity
	$$ H_j^k = h(j,j+1,j+2, \ldots ,k-1) + h(j+1,j+2,j+3, \ldots ,k). $$
\end{defn}

\begin{prop}
\label{prop:H=sum(H)}
For any three integers $i,j,k$ such that $1 \leq i < j$ and $j+1 < k \leq n$.  
	$$ H_i^k \leq H_i^j + H_{j+1}^k. $$
\end{prop}
\begin{proof}
We add up the following inequalities:
	\begin{align*}
	h(i,i+1, \ldots , k-1) 
	& = h(i, \ldots ,j-1,j+1, \ldots , k-1) \\
	& \leq  h(i, \ldots ,j-1) + h(j+1 \ldots , k-1), \\ \mbox{and} \quad  
	h(i+1,i+2, \ldots , k) 
	& = h(i+1, \ldots ,j,j+2, \ldots , k) \\
	& \leq  h(i+1, \ldots ,j) + h(j+2, \ldots , k).
	\end{align*}
\end{proof}

\begin{lem}
\label{lem:H(X),h(i)}
If $\Pr$ is a non-trivial protocol on $C_n$ with $s \geq 2$ colours and $X\in_u \fix(\Pr)$, then, for all $i$, 
	$$\log_s \fix(\Pr) = H(X), \;\; h(i) \leq 1.$$ 
\end{lem}
\begin{proof}
The entropy of any random variable over a finite domain is maximised when the variable is uniformly distributed. Therefore, $h(i) = H(X_i) \leq H(U)$ where $U$ is a random variable uniformly distributed over $\Z_s$. Hence, 
	$$ h(i) \leq H(U) = - \sum \frac{1}{s} \log_s \frac{1}{s} = 1. $$
The variable $X$ is uniformly distributed over $\Fix(\Pr)$. Therefore, 
	$$ H(X) = - \sum \frac{1}{\fix(\Pr)} \log_s \frac{1}{\fix(\Pr)} = \log_s \fix(\Pr). $$
\end{proof}

\begin{lem}
\label{lem:H_j^k<sum(h(i))}
If $\Pr$ is a non-trivial protocol on $C_n$ with $s \geq 2$ colours and $X\in_u \fix(\Pr)$, then 
	$$ H_j^k \leq \sum_{i=j}^{k} H(X_i) ,$$
for any $j \leq 1$ and $j+3 \leq k \leq n$. 
\end{lem}
\begin{proof}
We prove this by induction on $(k-j)$. Recall that $h(i_1,i_2,i_3, \ldots ) = H(X_{i_1},X_{i_2},X_{i_3}, \ldots )$.
	\begin{itemize}
	\item {\bf Base case: }$k=j+3$. Since $X_{j+1} = f_{j+1}(X_j,X_{j+2})$ and $X_{j+2} = f_{j+2}(X_{j+1},X_{j+3})$ we have 
		\begin{align*}
		h(j,j+1,j+2) = h(j,j+2) & \leq h(j)+h(j+2) \\ \mbox{and} \qquad 
		h(j+1,j+2,j+3) = h(j+1,j+3) & \leq h(j+1)+h(j+3) ,
		\end{align*}
	respectively. Adding these together yields:
		$$ H_j^{j+3} = h(j,j+1,j+2) + h(j+1,j+2,j+3) \leq h(j)+h(j+1)+h(j+2)+h(j+3). $$
	\item {\bf Inductive step:} $k \geq j+4$. Since $X_{k-1} = f_{k-1}(X_{k-2},X_k)$ we have 
		\begin{align*}
		h(j+1,j+2, \ldots ,k) 
		& = h(j+1,j+2, \ldots ,k-2,k) \\
		& \leq h(j+1,j+2, \ldots ,k-2)+h(k). 
		\end{align*}
	By Proposition~\ref{prop:known_entropy_results}, 
	$ I(X_j ; X_{k-1} | X_{j+1},X_{j+2}, \ldots ,X_{k-2})\geq 0$. Adding these together yields  
		$$ H_j^k \leq H_j^{k-1} + h(k). $$
	This completes the proof. 
	\end{itemize}
\end{proof}

\begin{lem}
\label{lem:H(X)=sum(H)}
Let $\Pr$ be a non-trivial protocol on $C_n$ with $s \geq 2$ colours and let $X\in_u \fix(\Pr)$. Suppose $1 = d(1),d(2),d(3), \ldots ,d(k) = n$ is a sequence of positive integers with $k \geq 2$. If $d(i+1) \geq d(i)+2$ for all $i$, then 
	$$ 2\log_s\fix(\Pr) = H_{d(1)}^{d(2)} + H_{d(2)+1}^{d(3)} + \cdots + H_{d(k-1)+1}^{d(k)}. $$
\end{lem}
\begin{proof} 
We proceed by induction on $k$. 
\begin{itemize}
\item {\bf Base case:} $k=2$. Since $X_1 = f_1(X_n,X_2)$ and $X_n = f_n(X_{n-1},X_1)$, we have 
	$$ H(X) = h(2,3,4, \ldots ,n) \qquad \mbox{and} \qquad H(X) = h(1,2,3, \ldots ,n-1) ,$$
respectively. Adding these together gives $H_1^n = 2H(X) = 2\fix(\Pr)$.
\item {\bf Inductive step.} By Proposition \ref{prop:H=sum(H)}, for any $d(k-1)+2 \leq d(k) \leq n-2$, we have 
	$$ H_{d(k-1)+1}^{n} = H_{d(k-1)+1}^{d(k)} + H_{d(k)+1}^{n}. $$
\end{itemize}
\end{proof}

\section{The Fractional-Clique-Partition Protocol}
\label{sec:frac}

In this section, we define the fractional-clique-partition protocol, $\Pr_{fcp}$, on odd cycles $C_n$ with $s\geq 2$ colours. Theorem \ref{thm:n/2} appears in \cite{Christofides2011guessing} and serves as a good upper bound for any $n \geq 4$ and all numbers of colours. 

\begin{defn}[Factorization bijection]
\label{defn:phi_and_psi}
It is easy to see that for any factorization $ab=s$, there exists a bijection between $\Z_s$ and $\Z_a \times \Z_b$. Let $\phi(z) \times \psi(z)$ be such a bijection. For ease of notation, $a$ and $b$ are assumed to be given in context. Let $\pi$ be the inverse of this bijection, so that $\pi(\phi(z),\psi(z))=z$ for all $z \in \Z_s$. 
\end{defn}

\begin{defn}[Fractional-clique-partition protocol]
\label{defn:P_fcp}
Let $n \geq 3$ be an odd integer, let $s$ be a positive integer, let $a$ be the greatest factor of $s$ less than or equal to $\sqrt{s}$ and let $b = s/a$. For any colouring $c = (c_1,c_2, \ldots ,c_n) \in \Z_s^n$, let $\phi(c_i)$ and $\psi(c_i)$ be referred to as the first and second coordinates respectively of vertex $v_i$. The \emph{fractional-clique-partition protocol} is the protocol $\Pr_{fcp} = (f_1,f_2, \ldots ,f_n)$ on $C_n$ defined by:
	\begin{align*}
	f_i(c_{i-1},c_{i+1}) & = \pi \big( \phi(c_{i-1}) , \psi(c_{i+1}) \big) & \mbox{for } i=2,4,6, \ldots ,n-1 \\
	f_i(c_{i-1},c_{i+1}) & = \pi \big( \phi(c_{i+1}) , \psi(c_{i-1}) \big) & \mbox{for } i=3,5,7, \ldots ,n-2 \\
	f_1(c_n,c_2) & = \pi \big( \phi(c_2) , \phi(c_n) \big) & \mbox{and} \\
	f_n(c_{n-1},c_1) & = \pi \big( \psi(c_1) (\mbox{mod } a), \psi(c_{n-1}) \big).
	\end{align*}
\end{defn}

Informally, vertices $v_{2k-1}$ and $v_{2k}$ are copying each others first coordinate and vertices $v_{2k}$ and $v_{2k+1}$ are copying each others second coordinate (for $k = 1,2,3, \ldots , (n-1)/2$). Additionally, the second coordinate of vertex $v_1$ and the first coordinate of vertex $v_n$ copy each other as much as possible - whenever the second coordinate of vertex $v_1$ is less than $a$. An example of $\Pr_{fcp}$ on $C_7$ is illustrated in Figure \ref{fig:57cycle}.

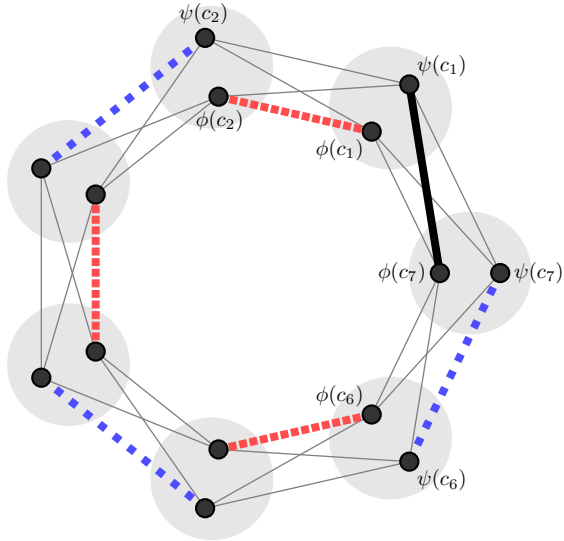
\begin{figure}
\begin{center}
\begin{tikzpicture}[thick,scale=0.8, every node/.style={scale=0.8}]
\draw[fill=black!10,black!10] (0*360/7:3.5) circle (1);
\draw[fill=black!10,black!10] (1*360/7:3.5) circle (1);
\draw[fill=black!10,black!10] (2*360/7:3.5) circle (1);
\draw[fill=black!10,black!10] (3*360/7:3.5) circle (1);
\draw[fill=black!10,black!10] (4*360/7:3.5) circle (1);
\draw[fill=black!10,black!10] (5*360/7:3.5) circle (1);
\draw[fill=black!10,black!10] (6*360/7:3.5) circle (1);
\draw[black!50,line width=.5pt] (0*360/7:3) -- (1*360/7:3);
\draw[black!50,line width=.5pt] (0*360/7:4) -- (1*360/7:3);
\draw[black!50,line width=.5pt] (0*360/7:4) -- (1*360/7:4);
\draw[black!100,line width=3pt] (0*360/7:3) -- (1*360/7:4);
\draw[black!50,line width=.5pt] (1*360/7:3) -- (2*360/7:4);
\draw[black!50,line width=.5pt] (1*360/7:4) -- (2*360/7:3);
\draw[black!50,line width=.5pt] (1*360/7:4) -- (2*360/7:4);
\draw[densely dotted,red!70,line width=3pt] (1*360/7:3) -- (2*360/7:3);
\draw[black!50,line width=.5pt] (2*360/7:3) -- (3*360/7:4);
\draw[black!50,line width=.5pt] (2*360/7:4) -- (3*360/7:3);
\draw[black!50,line width=.5pt] (2*360/7:3) -- (3*360/7:3);
\draw[loosely dotted,blue!70,line width=3pt] (2*360/7:4) -- (3*360/7:4);
\draw[black!50,line width=.5pt] (3*360/7:3) -- (4*360/7:4);
\draw[black!50,line width=.5pt] (3*360/7:4) -- (4*360/7:3);
\draw[black!50,line width=.5pt] (3*360/7:4) -- (4*360/7:4);
\draw[densely dotted,red!70,line width=3pt] (3*360/7:3) -- (4*360/7:3);
\draw[black!50,line width=.5pt] (4*360/7:3) -- (5*360/7:4);
\draw[black!50,line width=.5pt] (4*360/7:4) -- (5*360/7:3);
\draw[black!50,line width=.5pt] (4*360/7:3) -- (5*360/7:3);
\draw[loosely dotted,blue!70,line width=3pt] (4*360/7:4) -- (5*360/7:4);
\draw[black!50,line width=.5pt] (5*360/7:3) -- (6*360/7:4);
\draw[black!50,line width=.5pt] (5*360/7:4) -- (6*360/7:3);
\draw[black!50,line width=.5pt] (5*360/7:4) -- (6*360/7:4);
\draw[densely dotted,red!70,line width=3pt] (5*360/7:3) -- (6*360/7:3);
\draw[black!50,line width=.5pt] (6*360/7:3) -- (7*360/7:4);
\draw[black!50,line width=.5pt] (6*360/7:4) -- (7*360/7:3);
\draw[black!50,line width=.5pt] (6*360/7:3) -- (7*360/7:3);
\draw[loosely dotted,blue!70,line width=3pt] (6*360/7:4) -- (7*360/7:4);
\draw[fill=black!80] (0*360/7:3) circle (.15);
\draw[fill=black!80] (1*360/7:3) circle (.15);
\draw[fill=black!80] (2*360/7:3) circle (.15);
\draw[fill=black!80] (3*360/7:3) circle (.15);
\draw[fill=black!80] (4*360/7:3) circle (.15);
\draw[fill=black!80] (5*360/7:3) circle (.15);
\draw[fill=black!80] (6*360/7:3) circle (.15);
\draw[fill=black!80] (0*360/7:4) circle (.15);
\draw[fill=black!80] (1*360/7:4) circle (.15);
\draw[fill=black!80] (2*360/7:4) circle (.15);
\draw[fill=black!80] (3*360/7:4) circle (.15);
\draw[fill=black!80] (4*360/7:4) circle (.15);
\draw[fill=black!80] (5*360/7:4) circle (.15);
\draw[fill=black!80] (6*360/7:4) circle (.15);
\node [right] at (0*360/7:4.1) {$\psi(c_7)$};
\node [left] at (0*360/7:2.9) {$\phi(c_7)$};
\node [above right] at (1*360/7:4) {$\psi(c_1)$};
\node [below left] at (1*360/7:3) {$\phi(c_1)$};
\node [above] at (2*360/7:4.1) {$\psi(c_2)$};
\node [below] at (2*360/7:2.9) {$\phi(c_2)$};
\node [below right] at (6*360/7:4) {$\psi(c_6)$};
\node [above left] at (6*360/7:3) {$\phi(c_6)$};
\end{tikzpicture}
\caption[]{The protocol $\Pr_{fcp}$ on $C_7$ with $s=ab$ colours, where $a<b$. Each vertex $v_i$ is subdivided into two nodes representing the first and second components ($\phi(c_i)$ and $\psi(c_i)$, respectively). The red edges (\tikz[baseline=-0.5ex]{ \draw[densely dotted,red!70,line width=3pt] (0,0) -- (.5,0); }) represent pairs of first-components that are copying each other. The blue edges (\tikz[baseline=-0.5ex]{ \draw[loosely dotted,blue!70,line width=3pt] (0,0) -- (.5,0); }) represent pairs of second-components that are copying each other. The black edge (\tikz[baseline=-0.5ex]{ \draw[black!100,line width=3pt] (0,0) -- (.5,0); }) joins a first-component ($\phi(c_n)$) and a second-component ($\psi(c_1)$) which are copying each other as much as possible. For a colouring $c \in \Fix(\Pr_{fcp})$ on $C_7$, there are $a$ different choices for each red edge, $b$ different choices for each blue edge and $a$ different choices for the black edge. Therefore, $\fix(\Pr_{fcp}) = a^4b^3=as^3$ for $n=7$.}
\label{fig:57cycle}
\end{center}
\end{figure}

\begin{prop}
\label{prop:fix_Ffcp}
For a given integer $s \geq 2$ and odd integer $n \geq 3$, if $a$ is the greatest factor of $s$ less than or equal to $\sqrt{s}$, then we have $\fix(\Pr_{fcp}) = as^{(n-1)/2}$. 
\end{prop}
\begin{proof}
Let $n = 2k+1$. We count the number of colourings of $C_n$ for which the protocol $\Pr_{fcp}$ guesses correctly. For any colouring $c \in \Fix(\Pr_{fcp})$, there are $k$ pairs of vertices copying each other's first coordinates and there are $a$ different choices for $\phi$ for each pair. Similarly, for each of the $k$ pairs of vertices copying each other's second coordinates, there are $b$ different choices for $\psi$. This yields $a^kb^k$ possibilities. Additionally, the first coordinate of vertex $v_n$ must equal the second coordinate of vertex $v_1$, for which there are $a$ possible colours. Multiplying these together yields 
	$$ \fix(\Pr_{fcp}) = a^{k+1}b^k = as^{(n-1)/2}. $$
\end{proof}

\begin{thm}
\label{thm:n/2} \cite{Christofides2011guessing}
For any integer $n \geq 4$, we have $\displaystyle{\gn(C_n,s) \leq \tfrac{n}{2}}$, with equality only if for any optimal protocol, $\Pr$ the following is satisfied. If $X \in_u \Fix(\Pr)$ then $H(X_i)=1$ for all $i$. 
\end{thm}
\begin{proof}
Let $\Pr$ be an optimal protocol on $C_n$ with $s$ colours. By Lemmas~\ref{lem:H(X),h(i)},~\ref{lem:H_j^k<sum(h(i))} and ~\ref{lem:H(X)=sum(H)}, we have 
	$$ \gn(C_n,s) = \log_s\fix(\Pr) = H(X) = \tfrac{1}{2} H_1^n \leq \tfrac{1}{2} \sum_{i=1}^n h(i) \leq \frac{n}{2}. $$
If $\gn(C_n,s)=n/2$, then we must have equality throughout, which means that $h(i)=1$ for all $i$.
\end{proof}

Theorem~\ref{thm:n/2} appears in \cite{Christofides2011guessing}. This same paper also shows that the limit of $\gn(C_n,s) \to n/2$ as $s \rightarrow \infty$. We give a bound on the rate convergence to this limit in Theorem~\ref{thm:round_up}. 

\begin{thm}
\label{thm:round_up}
If $n$ is odd and $s = m^2-t$ for integers $m$ and $t \geq 0$ then there exists a protocol $\Pr$ on $C_n$ with $s$ colours such that 
	$$ \fix(\Pr) \geq s^{n/2} \left( 1 - \frac{tn}{s} \right). $$
\end{thm}
\begin{proof}
Consider the protocol $\Pr^\prime = \Pr_{fcp}$ on $C_n$ with $s^\prime = m^2$ colours and let $X^\prime \in_u \Fix(\Pr^\prime)$. By Theorem~\ref{thm:n/2}, we must have $H(X_i^\prime)=1$ and therefore $X_i^\prime$ is uniformly distributed over $\Z_{s^\prime}$ for all $i$. By the union bound, 
	$$ \P\left( X_i^\prime<s \; \forall \: i \right) \geq 1 - \sum_{i=1}^n \P(X_i^\prime \geq s) = 1 - \sum_{i=1}^n \frac{t}{m^2} = 1 - \frac{tn}{m^2}. $$
Now, let $\Pr$ be a protocol on $C_n$ with $s$ colours such that $c \in \Fix(\Pr)$ for all colourings $c \in \Fix(\Pr^\prime)$ such that $c_i < s$ for all $i$ (such a protocol must exist by Proposition~\ref{prop:monotone}). For this protocol,  
	\begin{align*} \fix(\Pr) 
	& \geq \fix(\Pr^\prime) \: \P\left( X_i^\prime<s \; \forall \: i \right) \\
	& \geq \fix(\Pr^\prime) \left( 1 - \frac{tn}{m^2} \right) \\
	& = (s+t)^{n/2} \left( 1 - tn(s+t)^{-1} \right) \\ 
	& \geq s^{n/2} \left( 1 - \frac{tn}{s} \right).
	\end{align*}
\end{proof}
\begin{cor}
If $n \geq 4$ then $\gn(C_n,s) = \frac{n}{2} - \O\left(\frac{n}{\sqrt{s}\log_es} \right)$ as $s \rightarrow \infty$.
\end{cor}
\begin{proof}
For even $n$ we have $\gn(C_n,s) = \frac{n}{2}$. For odd $n$, let $m$ be the smallest positive integer such that $m^2 \geq s$. This gives $t = m^2-s = \O(\sqrt{s})$. If $\Pr$ is the protocol constructed in Theorem~\ref{thm:round_up}, then 
	$$ \gn(C_n,s) \geq \log_s\fix(\Pr) \geq \frac{n}{2} + \log_s \left( 1 - \frac{tn}{s} \right) = \frac{n}{2} - \O\left( \frac{n}{\sqrt{s}\log_es} \right). $$
\end{proof}

\section{Entropy Results}
\label{sec:entropy}

The bounds in Theorem~\ref{thm:round_up} are only useful when $n$ is small relative to $s$. In contrast, the purpose of the results in this section is to establish Lemma~\ref{lem:many_non-semi-perfect_functions}, which in turn will be used to prove Theorem~\ref{thm:n_large} which only applies when $n$ is large relative to $s$. To help orientate the reader through this section (and the next), Figure~\ref{fig:proof_digraph} shows which results are used to prove other results.

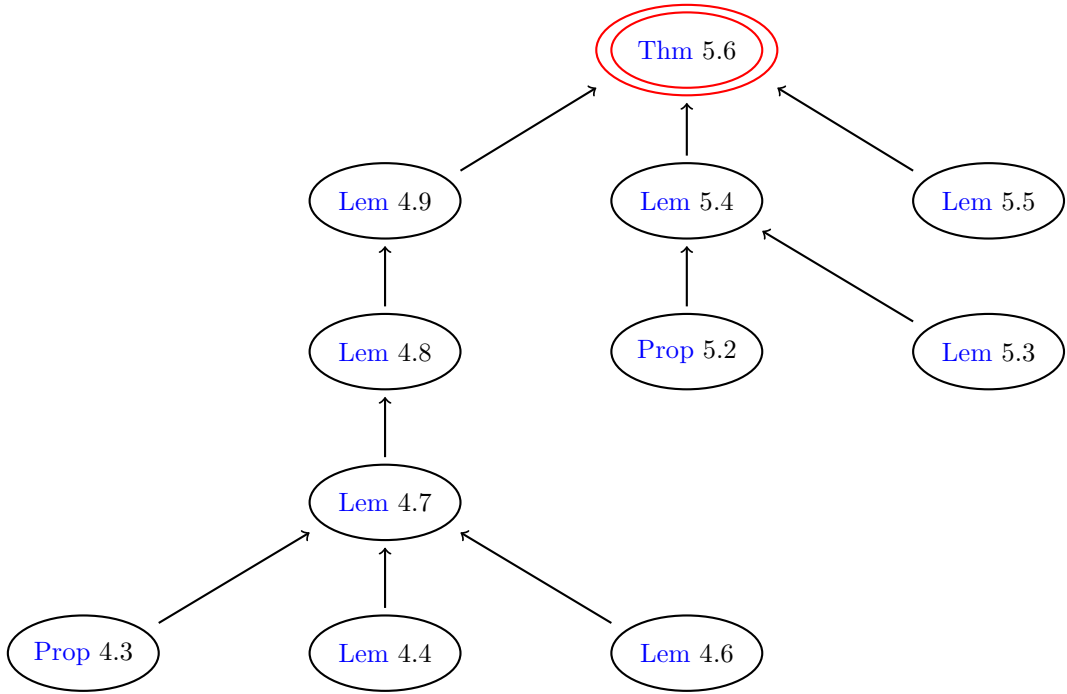
\begin{figure}
\begin{center}
\begin{tikzpicture}[thick]
\draw (0,0) ellipse (1.0 and 0.5);
\node [blue] at (0,0) {Prop~\ref{prop:delta}};
\draw [->] (1,0.4) --(3,1.6);
\draw (4,0) ellipse (1.0 and 0.5);
\node [blue] at (4,0) {Lem~\ref{lem:alpha}};
\draw [->] (4,0.6) --(4,1.4);
\draw (8,0) ellipse (1.0 and 0.5);
\node [blue] at (8,0) {Lem~\ref{lem:shannon}};
\draw [->] (7,0.4) --(5,1.6);
\draw (4,2) ellipse (1.0 and 0.5);
\node [blue] at (4,2) {Lem~\ref{lem:one_non-semi-perfect_function}};
\draw [->] (4,2.6) --(4,3.4);
\draw (4,4) ellipse (1.0 and 0.5);
\node [blue] at (4,4) {Lem~\ref{lem:at_least_one_non-semi-perfect_function}};
\draw [->] (4,4.6) --(4,5.4);
\draw (4,6) ellipse (1.0 and 0.5);
\node [blue] at (4,6) {Lem~\ref{lem:many_non-semi-perfect_functions}};
\draw [->] (5,6.4) --(6.8,7.5);
\draw [red] (8,8) ellipse (1.0 and 0.5);
\draw [red] (8,8) ellipse (1.2 and 0.6);
\node [blue] at (8,8) {Thm~\ref{thm:n_large}};
\draw [->] (11,6.4) --(9.2,7.5);
\draw (12,6) ellipse (1.0 and 0.5);
\node [blue] at (12,6) {Lem~\ref{lem:one_perfect_function}};
\draw [->] (8,6.6) --(8,7.3);
\draw (8,6) ellipse (1.0 and 0.5);
\node [blue] at (8,6) {Lem~\ref{lem:three_semi-perfect_functions}};
\draw [->] (11,4.4) --(9,5.6);
\draw (12,4) ellipse (1.0 and 0.5);
\node [blue] at (12,4) {Lem~\ref{lem:two_semi-perfect_functions}};
\draw [->] (8,4.6) --(8,5.4);
\draw (8,4) ellipse (1.0 and 0.5);
\node [blue] at (8,4) {Prop~\ref{prop:LR=s}};
\end{tikzpicture}
\caption{The structure of Sections~\ref{sec:entropy} and~\ref{sec:n_large}. An arrow $A \rightarrow B$ indicates that $A$ is used in the proof of $B$.}
\label{fig:proof_digraph}
\end{center}
\end{figure}

\begin{defn}[Flat function, semi-perfect function]
\label{defn:semi-perfect_function}
For any $z \in \Z_s$ and for any function $f : \Z_s^2 \rightarrow \Z_s$ let $f^{-1}(z) =  \{ (x,y) \: | \: f(x,y)=z \}$. The function $f$ is called \emph{flat} if and only if $|f^{-1}(z)| = s$ for all $z$. Let $U = (U_1,U_2) \in_u \Z_s^2$. A \emph{semi-perfect} function, $f$, is any flat function such that the $U_1$ and $U_2$ are conditionally independent given $f(U)$ (Definition~\ref{defn:entropy}), \emph{i.e.} 
	$$ I(U_1;U_2 \: | \: f(U)) = 0. $$
\end{defn}

\begin{defn}[$\mathbf{(k,\boldsymbol{\epsilon})}$-uniform]
\label{defn:(k,e)-uniform} 
For any positive integer $k$ and any $\epsilon >0$, a random variable $Y$ is called $(k,\epsilon)$-uniform if $Y$ takes values in a finite set $\mathcal{Y}$ with $|\mathcal{Y}| = k$ and, for any $y \in \mathcal{Y}$,
	$$ \left| \P(Y=y) - \frac{1}{k} \right| \leq \epsilon .$$

\end{defn}

\begin{prop}
\label{prop:delta}
For any integer $k \geq 2$, any integer $s \geq 2$ and any $\epsilon >0$, there exists $\delta >0$ such that, for any random variable $Y$ which takes $k$ distinct values, if $H(Y)$ is the entropy of $Y$ (base $s$), then 
	$$ H(Y) \geq \log_s k - \delta \qquad \implies \qquad Y \mbox{ is $(k,\epsilon)$-uniform.} $$
\end{prop}
\begin{proof}
For each $k$, it suffices to show this for all small enough $\epsilon$. Assume $7k\epsilon < 1$. We prove the contrapositive:
\begin{itemize}
\item Suppose that $\P(Y=y) \geq \frac{1}{k} + \epsilon$ for at least one value $y$. Entropy is greatest when $Y$ is as uniformly distributed as possible. Therefore,
	\begin{align*} H(Y) 
	& = - \sum_i \P(Y=i) \log_s \P(Y=i) \\
	& \leq - \left( \tfrac{1}{k} + \epsilon \right) \log_s \left( \tfrac{1}{k}+\epsilon \right) 
	- (k-1) \left( \tfrac{1}{k} - \tfrac{\epsilon}{k-1} \right) \log_s \left( \tfrac{1}{k} - \tfrac{\epsilon}{k-1} \right) \\
	& = \log_s k - \left( \tfrac{1}{k} + \epsilon \right) \log_s( 1 + k\epsilon )
	- \left( \tfrac{k-1}{k} - \epsilon \right) \log_s \left( 1 - \tfrac{k\epsilon}{k-1} \right). 
\intertext{Since $0 < k\epsilon < \frac{1}{7}$, we can use the identity, $-\log_s(1-\gamma) \leq (\gamma + \frac{5}{9}\gamma^2)\log_es$ (valid for $|\gamma| \leq 1/7$), to simplify this expression.}
	H(Y) 
	& \leq \log_sk - \frac{k\epsilon^2}{9} \left( 4 - 5k\epsilon + \tfrac{4}{k-1} + \tfrac{5k\epsilon}{(k-1)^2} \right)\log_es \\
	& \leq \log_sk - \frac{k\epsilon^2}{3}\log_es.
	\end{align*}
\item Now suppose $\P(Y=y) \leq \frac{1}{k} - \epsilon$ for at least one value $y$. The entropy would be greatest when $Y$ is as uniformly distributed as possible. Therefore, 
	\begin{align*} H(Y) 
	& = - \sum_i \P(Y=i) \log_s \P(Y=i) \\
	& \leq - \left( \tfrac{1}{k} - \epsilon \right) \log_s \left( \tfrac{1}{k} - \epsilon \right) 
	-(k-1) \left( \tfrac{1}{k} + \tfrac{\epsilon}{k-1} \right) \log_s \left( \tfrac{1}{k} - \tfrac{\epsilon}{k-1} \right) \\
	& = \log_sk - \left( \tfrac{1}{k} - \epsilon \right) \log_s ( 1 - k\epsilon )
	 - \left( \tfrac{k-1}{k} + \epsilon \right) \log_s \left( 1 + \tfrac{k\epsilon}{k-1} \right). 
\intertext{We can use the identity, $-\log_s(1-\gamma) \leq (\gamma + \frac{5}{9}\gamma^2)\log_es$, again to simplify this expression. }
	H(Y) 
	& \leq \log_sk - \tfrac{k\epsilon^2}{9} \left( 4 - 5k\epsilon + \tfrac{4}{k-1} - \tfrac{5k\epsilon}{(k-1)^2} \right) \log_es\\
	& \leq \log_sk - \frac{k\epsilon^2}{3}\log_es.
	\end{align*}
\end{itemize}
In either case, $H(Y) < \log_s k - \delta$ for any $\delta < \frac{k\epsilon^2}{3}\log_es$. 
\end{proof}

\begin{lem}
\label{lem:alpha} 
For any integer $s \geq 2$, there exists positive constant $\epsilon = \epsilon(s)$ that satisfies the following property. 
For any non semi-perfect function $f : \Z_s^2 \rightarrow \Z_s$ and for any three $(s,\epsilon)$-uniform random variables $Y_1,Y_2,Y_3$ over $\Z_s$ satisfying $Y_2 = f(Y_1,Y_3)$, if $(Y_1,Y_3)$ is $(s^2,\epsilon)$-uniform, then 
  $$ I(Y_1;Y_3|Y_2) \geq \tfrac{1}{2} \min \big\{ I(U_1;U_2|f(U)) \; \big| \; f \mbox{ is a flat but not semi-perfect} \big\} = \delta_1,$$
where $U = (U_1,U_2) \in_u \Z_s^2$.
\end{lem}
\begin{proof}
The value $\delta_1=\delta_1(s) = \tfrac{1}{2} \min \big\{ I(U_1;U_2|f(U)) \; \big| \; f \mbox{ is a flat but not semi-perfect} \big\}$ is well-defined for any $s \geq 2$, because there are only a finite number of functions $f : \Z_s^2 \rightarrow \Z_s$, and at least one of them is flat and not semi-perfect hence we can take the minimum of these. For example, the function $f(x,y) = x+y$ (mod $s$) is flat but not semi-perfect. 
First, let $\epsilon < \frac{1}{s^2(s+2)}$, so that 
	$$ \frac{1}{s^2} - (s-1)\epsilon > \frac{1}{s^2} - (s+1)\epsilon > \epsilon. $$ 
We show that $f$ is flat by contradiction. Since $(Y_1,Y_3)$ is $(s^2,\epsilon)$-uniform: 
\begin{itemize}
\item If $| f^{-1}(z) | \geq s+1$ then 
	$$ \P(Y_2=z) = \P( (Y_1,Y_3) \in f^{-1}(z)) \geq (s+1) \left( \tfrac{1}{s^2} - \epsilon \right) 
	= \tfrac{1}{s} + \left( \tfrac{1}{s^2} - (s+1)\epsilon \right) > \tfrac{1}{s} + \epsilon. $$
\item If $| f^{-1}(z) | \leq s-1$ then 
	$$ \P(Y_2=z) = \P( (Y_1,Y_3) \in f^{-1}(z)) \leq (s-1) \left( \tfrac{1}{s^2} + \epsilon \right) 
	= \tfrac{1}{s} - \left( \tfrac{1}{s^2} - (s-1)\epsilon \right) < \tfrac{1}{s} - \epsilon. $$
\end{itemize}
Both cases contradict the assumption that $Y_2$ is $(s,\epsilon)$-uniform. Therefore $f$ is a flat function and so $$I(U_1;U_2|f(U)) \geq 2\delta_1.$$

Moreover, since $(Y_1,Y_3)$ is $(s^2,\epsilon)$-uniform, then $U$ and $(Y_1,Y_3)$ differ in distribution by less than $\epsilon$. Since mutual information is continuous, we can choose $\epsilon$ small enough so that 
	$$ \big| I(U_1;U_2|f(U)) - I(Y_1;Y_3|Y_2) \big| \leq \delta_1. $$
Then, by the triangle inequality, 
$I(Y_1;Y_3|Y_2) \geq \delta_1$.

\end{proof}

\begin{defn}
\label{defn:epsilon,delta} 
From now on, for any integer $s \geq 2$, let $\epsilon = \epsilon(s) >0$ be chosen small enough so that $\epsilon \leq \frac{1}{s^2(2s+1)}$ and $\epsilon$ satisfies Lemma~\ref{lem:alpha}. Then let $\delta_2 = \delta_2(s) >0$ be chosen small enough to satisfy Proposition~\ref{prop:delta} for both $k=s$ and $k=s^2$ for this value $\epsilon$. Then, with $\delta_1$ as defined in Lemma \ref{lem:alpha}, let $\delta= \min (\delta_1,\delta_2)$.
\end{defn}

\begin{lem}
\label{lem:shannon}
Let $n \geq 5$ be an integer and let $\Pr$ be any non-trivial protocol on $C_n$ with $s \geq 2$ colours. The random variables $X_1,X_2,X_3,X_4,X_5$ (Definition~\ref{defn:C_n}) satisfy:
	$$ H_1^5 \leq 3+h(2,4) - I(X_2;X_4|X_3). $$
\end{lem}
\begin{proof}
By Lemma~\ref{lem:H(X),h(i)}, it suffices to show $H_1^5 \leq h(1)+h(3)+h(5)+h(2,4) - I(X_2;X_4|X_3)$. By Shannon's Inequality (Proposition~\ref{prop:known_entropy_results}) we have:
	\begin{align}
	h(2,3,4) + h(3) & = h(2,3) + h(3,4) - I(X_2;X_4|X_3), \label{eq:mutual_information} \\
	h(1,2,3,4) + h(2,3) & \leq h(1,2,3) + h(2,3,4), \label{eq:shannon12} \\
	\mbox{and} \quad h(2,3,4,5) + h(3,4) & \leq h(2,3,4) + h(3,4,5). \label{eq:shannon23} \\
\intertext{Also, since $X_i = f_i(X_{i-1},X_{i+1})$ for $i=2,3,4$ respectively we have:}
	h(1,2,3) &= h(1,3)  \leq h(1) + h(3), \label{eq:func1} \\
	h(2,3,4) & = h(2,4), \label{eq:func2} \\
	\mbox{and} \quad h(3,4,5) &= h(3,5)  \leq h(3) + h(5). \label{eq:func3} 	
	\end{align}
The required result is the sum of equations (\ref{eq:mutual_information}), (\ref{eq:shannon12}), (\ref{eq:shannon23}), (\ref{eq:func1}), (\ref{eq:func2}) and (\ref{eq:func3}).
\end{proof}

\begin{lem}
\label{lem:one_non-semi-perfect_function}
Let $n \geq 5$ be an integer and let $\Pr = (f_1,f_2, \ldots ,f_n)$ be a non-trivial protocal on $C_n$ with $s \geq 2$ colours and let $X\in_u \fix(\Pr)$. For any $j$, if $f_{j+2}$ is not semi-perfect or $(X_{j+1},X_{j+3})$ is not $(s^2,\epsilon)$-uniform then $H_j^{j+4} \leq 5 - \delta$, for $\delta$ as in Definition \ref{defn:epsilon,delta}. 
\end{lem}
\begin{proof}
Without loss of generality let $j=1$. There are $3$ cases.
\begin{itemize}
\item If, for any $i \in \{ 1,2,3,4,5 \}$, the variable $X_i$ is not $(s,\epsilon)$-uniform, then $h(i) \leq 1 - \delta_2$ (Proposition~\ref{prop:delta}). In this case, by Lemma~\ref{lem:H_j^k<sum(h(i))},
	$$ H_1^5 \leq \sum_{i=1}^5 h(i) \leq 5 - \delta_2. $$
\item If $(X_2,X_4)$ is not $(s^2,\epsilon)$-uniform, then $h(2,4) \leq 2-\delta_2$ (Proposition~\ref{prop:delta}). Therefore, by Lemma~\ref{lem:shannon}, we have 
	$$ H_1^5 \leq 3 + h(2,4) - I(X_2;X_4|X_3) \leq 5-\delta_2. $$
\item Otherwise, $X_2,X_3,X_4$ are each $(s,\epsilon)$-uniform and $(X_2,X_4)$ is $(s^2,\epsilon)$-uniform and $f_3$ is not semi-perfect. In this case, by Lemma~\ref{lem:alpha}, we have $I(X_{j+1};X_{j+3}|X_{j+2}) \geq \delta_1$. By Lemma~\ref{lem:shannon}, we have 
	$$ H_1^5 \leq 3 + h(2,4) - I(X_2;X_4|X_3) \leq 5-\delta_1. $$
\end{itemize}
In all cases, we have $H_1^5 \leq 5 - \delta$ because $\delta = \min\{ \delta_1 , \delta_2 \}$.
\end{proof}

\begin{lem}
\label{lem:at_least_one_non-semi-perfect_function}
Let $n \geq 7$ be an integer and let $\Pr = (f_1,f_2, \ldots ,f_n)$ a non-trivial protocol on $C_n$ with $s \geq 2$ colours and let $X\in_u \fix(\Pr)$. For any $j$, if any of $f_{j+2}$, $f_{j+3}$ or $f_{j+4}$ are not semi-perfect, or any of $(X_{j+1},X_{j+3})$, $(X_{j+2},X_{j+4})$ or $(X_{j+3},X_{j+5})$ are not $(s^2,\epsilon)$-uniform, then $H_{j}^{j+6} \leq 7 - \delta$.
\end{lem}
\begin{proof}
Without loss of generality let $j=1$. We treat each case individually, and use Lemma~\ref{lem:one_non-semi-perfect_function}.
\begin{itemize}
\item If $f_3$ is not semi-perfect or $(X_2,X_4)$ is not $(s^2,\epsilon)$-uniform then 
	\begin{align*} H_1^7 
	& = h(1,2,3,4,5,6) + h(2,3,4,5,6,7) \\
	& = h(1,2,3,4,6) + h(2,3,4,5,7) \\
	& \leq H_1^5 + h(6) + h(7) \\
	& \leq (5 - \delta) + 1 + 1.
	\end{align*}
\item If $f_4$ is not semi-perfect or $(X_3,X_5)$ is not $(s^2,\epsilon)$-uniform then 
	\begin{align*} H_1^7
	& = h(1,2,3,4,5,6) + h(2,3,4,5,6,7) \\
	& = h(1,3,4,5,6) + h(2,3,4,5,7) \\
	& \leq h(1) + H_2^6 + h(7) \\
	& \leq 1 + (5 - \delta) + 1.
	\end{align*}
\item If $f_5$ is not semi-perfect or $(X_4,X_6)$ is not $(s^2,\epsilon)$-uniform then 
	\begin{align*} H_1^7
	& = h(1,2,3,4,5,6) + h(2,3,4,5,6,7) \\
	& = h(1,3,4,5,6) + h(2,4,5,6,7) \\
	& \leq h(1) + h(2) + H_3^7 \\
	& \leq 1 + 1 + (5 - \delta).
	\end{align*}
\end{itemize}
\end{proof}

\begin{lem}
\label{lem:many_non-semi-perfect_functions}
Let $n \geq 7(\delta^{-1}+2)$. Suppose $\Pr = (f_1,f_2, \ldots ,f_n)$ is a non-trivial protocol on $C_n$ with $s \geq 2$ colours and let $X\in_u \fix(\Pr)$ such that, for each $j$, either
	\begin{itemize}
	\item at least one of $f_{j-1}$, $f_j$, $f_{j+1}$ is not semi-perfect, or 
	\item at least one of $(X_{j-2},X_j)$, $(X_{j-1},X_{j+1})$, $(X_j,X_{j+2})$ is not $(s^2,\epsilon)$-uniform,
	\end{itemize}
	  then $\fix(\Pr) < s^{(n-1)/2}$. 
\end{lem}
\begin{proof}
Let $m$ be an odd integer such that $m > \delta^{-1}$ and $7m \leq n$. By Lemma~\ref{lem:H(X)=sum(H)} and Lemma~\ref{lem:at_least_one_non-semi-perfect_function}, we have
	\begin{align*} 2H(X) 
	& \leq \sum_{j=0}^{m-1} H_{7j+1}^{7j+7} + \sum_{i=7m}^{n-1} h(i) \\
	& \leq m(7 - \delta) + (n-7m) \\
	& = n - m \delta.
	\end{align*}
Since $m > \delta^{-1}$, this means that $H(X) < \frac{n-1}{2}$. Therefore $\fix(\Pr) < s^{(n-1)/2}$ by Lemma~\ref{lem:H(X),h(i)}.
\end{proof}

\section{Guessing numbers of large odd cycles}
\label{sec:n_large}

In this section, we prove our main result in Theorem~\ref{thm:n_large}, which states that, for any given $s$, this fractional-clique-partition protocol is optimal on any large enough odd cycle.
\begin{defn}[Perfect function]
\label{defn:perfect_function}
For any function $f: \Z_s^2 \to \Z$, let $L(f,z)$ and $R(f,z)$ denote the subsets 
	\begin{align*} 
	L(f,z) & = \{ x \: | \: f(x,y)=z \mbox{ for some } y \} \\ \mbox{and} \quad  
	R(f,z) & = \{ y \: | \: f(x,y)=z \mbox{ for some } x \}.
	\end{align*}	
The function $f$ is called a \emph{perfect} function if it is semi-perfect and the cardinalities $|L(f,z)|$ and $|R(f,z)|$ do not depend on $z$, 
\emph{i.e.} if $|L(f,z)|=|L(f,z^\prime)|$ and $|R(f,z)|=|R(z^\prime)|$ for all $z,z^\prime \in \Z_s$.
\end{defn}

\begin{prop}
\label{prop:LR=s}
If $f$ is a semi-perfect function then for all $z \in \Z_s$ then
	$$ f^{-1}(z) = L(f,z) \times R(f,z). $$
Moreover $|L(f,z)||R(f,z)| = s$. 
\end{prop}
\begin{proof}
Let $U = (U_1,U_2) \in_u \Z_s^2$ and for a given $z$, let $L = L(f,z)$ and $R = R(f,z)$. Since $f$ is semi-perfect, we have $I(U_1,U_2 \; | \; f(U))=0$. Therefore, $U_1$ and $U_2$ are conditionally independent given $f(U)$. For any $x \in L$ and any $y \in R$, we must have
	\begin{align*}& \P(U_1=x \wedge U_2=y | f(U)=z) = \P(U_1=x | f(U)=z)\P(U_2=y | f(U)=z) > 0, \\
	&\mbox{and }f^{-1}(z) = L \times R.\end{align*} Furthermore, since $U_1$ and $U_2$ are independently uniformly distributed over $\Z_s$ and $U$ is uniformly distributed over $\Z_s^2$, we have
	$$ \frac{1}{s}  = \P( f(U)=z) = \P( U_1 \in L \wedge U_2 \in R) = \P( U_1 \in L) \times \P(U_2 \in R) = \frac{|L|}{s} \times \frac{|R|}{s}. $$
Therefore, $|L||R| = s$.
\end{proof}

\begin{lem}
\label{lem:two_semi-perfect_functions}
Let $s \geq 2$ be an integer, let $0 <\epsilon  \leq \frac{1}{s^2(2s+1)}$ be a constant. Let $\Pr = (f_1,f_2, \ldots ,f_n)$ be any non-trivial protocol on $C_n$ with $s$ colours and let $X \in_u \Fix(\Pr)$. If $f_1$ and $f_2$ are semi-perfect functions and $(X_0,X_2)$ and $(X_1,X_3)$ are $(s^2,\epsilon)$-uniform, then, for any $c_1,c_2 \in \Z_s$, we have 
	$$ | \{ c_0 | f_1(c_0,c_2)=c_1 \} | = | \{ c_3 | f_2(c_1,c_3)=c_2 \} |. $$
\end{lem}
\begin{proof} We proceed by contradiction. 
Let $S_0 = \{ c_0 | f_1(c_0,c_2)=c_1 \}$ and $S_3 = \{ c_3 | f_2(c_1,c_3)=c_2 \}$. Without loss of generality assume $|S_0|<|S_3|$ so since $|S_0| < s$ we must have $|S_3| > \left( 1 + \frac{1}{s} \right)|S_0|$. Now since $(X_0,X_2)$ is $(s^2,\epsilon)$-uniform,
	$$ \P( X_1=c_1 \wedge X_2=c_2) = \sum_{x \in S_0} \P\big( (X_0,X_2) = (x,c_2) \big) 
	\leq |S_0| \left( \frac{1}{s^2} + \epsilon \right). $$
Similarly, since $(X_1,X_3)$ is $(s^2,\epsilon)$-uniform, 
	$$ \P( X_1=c_1 \wedge X_2=c_2) = \sum_{x \in S_3}\P\big( (X_1,X_3) = (c_1,x) \big) 
	\geq |S_3| \left( \frac{1}{s^2} - \epsilon \right). $$
However, since $\epsilon \leq \frac{1}{s^2(2s+1)}$, this implies 
	$$ 1 + \frac{1}{s} < \frac{|S_3|}{|S_0|} \leq \frac{s^{-2}+\epsilon}{s^{-2}-\epsilon} 
	\leq \frac{\frac{1}{s^2}+\frac{1}{s^2(2s+1)}}{\frac{1}{s^2}-\frac{1}{s^2(2s+1)}} = 1 + \frac{1}{s}, $$
which is a contradiction.
\end{proof}

\begin{lem}
\label{lem:three_semi-perfect_functions}
Let $\Pr = (f_1,f_2, \ldots ,f_n)$ be a non-trivial protocol on $C_n$ with $s \geq 2$ colours, let $X \in_u \Fix(\Pr)$ and let $j$ be any index (indices taken modulo $n$). If $f_{j-1}$, $f_j$ and $f_{j+1}$ are semi-perfect functions and $(X_{j-2},X_j)$, $(X_{j-1},X_{j+1})$ and $(X_j,X_{j+2})$ are $(s^2,\epsilon)$-uniform, then $f_j$ is a perfect function.
\end{lem}
\begin{proof} We proceed by contradiction.
Without loss of generality, assume $j=0$ and fix $c_0,c_0^\prime \in \Z_s$ arbitrarily. Now choose $c_{-1},c_1\in \Z_s$ such that $f_0(c_{-1},c_1)=c_0$ and choose $c_{-1}^\prime,c_1^\prime \in \Z_s$ such that $f_0(c_{-1}^\prime ,c_1^\prime)=c_0^\prime$. Also let $c_0^{\prime\prime} = f_0(c_{-1}^\prime ,c_1)$. Now by Lemma~\ref{lem:two_semi-perfect_functions},
	\begin{align*}
	|L(f_0,c_0)| = | \{ x | f_0(x,c_1)=c_0 \} | = | \{ x | f_1(c_0,x)=c_1 \} | & = |R(f_1,c_1)| \\ \mbox{and} \qquad
	|L(f_0,c_0^{\prime\prime})| = | \{ x | f_0(x,c_1)=c_0^{\prime\prime} \} | = | \{ x | f_1(c_0^{\prime\prime},x)=c_1 \} | & = |R(f_1,c_1)|. 
	\end{align*}
Similarly
	\begin{align*}
	|R(f_0,c_0^{\prime\prime})| = | \{ x | f_0(c_{-1}^{\prime},x)=c_0^{\prime\prime} \} | = | \{ x | f_{-1}(x,c_0^{\prime\prime})=c_{-1}^{\prime} \} | & = |L(f_{-1},c_{-1}^{\prime})| \\ \mbox{and} \qquad
	|R(f_0,c_0^{\prime})| = | \{ x | f_0(c_{-1}^{\prime},x)=c_0^{\prime} \} | = | \{ x | f_{-1}(x,c_0^{\prime})=c_{-1}^{\prime} \} | & = |L(f_{-1},c_{-1}^{\prime})|.
	\end{align*} 

Recall that $|L(f_0,z)|\cdot |R(f_0,z)|=s$ for all $z\in \Z_s$ (Proposition~\ref{prop:LR=s}). Therefore, $|R(f_0,c_0^{\prime})|=|R(f_0,c_0^{\prime\prime})|$ if and only if $|L(f_0,c_0^{\prime})|=|L(f_0,c_0^{\prime\prime})|$. Hence, 
	$$ |L(f_0,c_0)| = |L(f_0,c_0^{\prime\prime})| = |L(f_0,c_0^{\prime})|. $$
Similarly, $|R(f_0,c_0)| = |R(f_0,c_0^\prime)|$ (for arbitrary $c_0,c_0^\prime \in \Z_s$) and therefore $f_0$ is a perfect function.
\end{proof}

\begin{lem}
\label{lem:one_perfect_function}
Let $\Pr = (f_1,f_2, \ldots ,f_n)$ be a non-trivial protocol on $C_n$ with $s \geq 2$ colours, such that $f_j$ is a perfect function for some $j$. Then $\fix(\Pr) \leq as^{(n-1)/2}$, where $a$ is the greatest factor of $s$ less than or equal to $\sqrt{s}$.
\end{lem}
\begin{proof}
Without loss of generality, assume $j=2$. Since $f_2$ is perfect, let $l = |L(f_2,z)|$ and $r=|R(f_2,z)|$. Without loss of generality, assume $l \leq r$ and therefore $l \leq a$. Then $X_2$ takes at most $s$ different values and $X_1$, conditioned on $X_2=z$ for any $z \in \Z_s$, takes at most $l$ different values. Therefore, the pair $(X_1,X_2)$ takes at most $ls$ different values in $\Z_s^2$ and $h(1,2) \leq \log_s(ls)$. We have
	\begin{align*}
	H(X) & = h(1,2,3, \ldots ,n) \\
	& = h(1,2,4,6, \ldots ,n-3,n-1) \\
	& \leq h(1,2) + \sum_{i=1}^{(n-3)/2} h(2i+2) \\
	& \leq \log_s(ls) + \frac{n-3}{2}. 
	\end{align*}
Therefore $\fix(\Pr) = s^{H(X)} \leq ls^{(n-1)/2} \leq as^{(n-1)/2}$.
\end{proof}

\begin{thm}
\label{thm:n_large}
For any integer $s \geq 2$, let $a$ be the greatest factor of $s$ less than or equal to $\sqrt{s}$. There exists some $N \in \mathbb{N}$ such that 
$$\gn (C_n,s)=\begin{cases} \frac{n}{2}, &\mbox{for even } n , \\  \frac{n-1}{2}+\log_s a, &\mbox{for odd }n>N,   \end{cases}$$
and $\Pr_{fcp}$ is an optimal protocol on $C_n$ with $s$ colours for any odd $n \geq N$.
\end{thm}
\begin{proof}
Let $\epsilon$ and $\delta$ be the values given in Definition~\ref{defn:epsilon,delta}, let $N = 7(\delta^{-1}+2)$ and let $\Pr = (f_1,f_2, \ldots ,f_n)$ be any non-trivial protocol on $C_n$ with $s$ colours. We have two cases:
\begin{description}
\item[Case one] 
	For all $j$, either:
	\begin{itemize} 
	\item at least one of the functions $f_{j-1}$, $f_j$ and $f_{j+1}$ is not semi-perfect or 
	\item at least one of $(X_{j-2},X_j)$, $(X_{j-1},X_{j+1})$, $(X_j,X_{j+2})$ is not $(s^2,\epsilon)$-uniform. 
	\end{itemize}
\item[Case two] 
	There exists some $j$ such that: 
	\begin{itemize}
	\item the functions $f_{j-1}$, $f_j$ and $f_{j+1}$ are all semi-perfect and 
	\item $(X_{j-2},X_j)$, $(X_{j-1},X_{j+1})$ and $(X_j,X_{j+2})$ are all $(s^2,\epsilon)$-uniform.
	\end{itemize}
\end{description}
For case one, we can conclude that $\fix(\Pr) \leq s^{(n-1)/2} \leq \fix(\Pr_{fcp})$ by Lemma~\ref{lem:many_non-semi-perfect_functions}. In case two, $f_j$ must be a perfect function (Lemma~\ref{lem:three_semi-perfect_functions}) and then $\fix(\Pr) \leq as^{(n-1)/2} = \fix(\Pr_{fcp})$ (Lemma~\ref{lem:one_perfect_function}). In either case, $\fix(\Pr_{fcp}) \geq \fix(\Pr)$. Hence $\Pr_{fcp}$ is optimal.
\end{proof}

\section{An application to index coding with side information}
\label{sec:defective}

In the problem of index coding with side information on a graph $G$, a sender aims communicate $n$ messages $c_1,c_2, \ldots ,c_n$ (where $c_i \in \Z_s$) to $n$ receivers $v_1,v_2, \ldots ,v_n$ (the vertices of $G$). Each receiver, $c_i$, knows $c_j$ in advance, for each $j$ such that $v_iv_j$ is an edge in $G$. The sender is required to broadcast a message to all receivers (the same message to all receivers) so that each receiver, $v_i$, can recover $c_i$. If $m$ is the smallest integer such that the sender can achieve this by broadcasting one of only $m$ different messages, then the \emph{information defect} \cite{Riis2007graph} of $G$ with $s$ colours is defined to be 
	$$ \b(G,s)  = \log_s(m). $$
The relationship between the guessing number and information defect of a graph is well known. 
Explicitly, let $\C _s(G)$ be the \emph{confusion graph} \cite{Alon2008broadcasting,bar2011index} (also known as the ``code graph'' \cite{Christofides2011guessing}), defined to have vertex set $\Z_s^n$, in which two vertices $c,c' \in \Z_s^n$ are adjacent if and only if for some $i \in [n]$, $c_i \neq c_i^\prime$ but for each $j$ such that $ij \in E(G)$ we have $c_j=c_j^\prime$. 
Intuitively $c,c^\prime \in Z_s^n$ are `confusable' (joined by an edge in the confusion graph) if there is no protocol $\Pr$, for the guessing game on $G$, such that both $c,c^\prime \in \Fix(\Pr)$ (\emph{i.e.} $c$ and $c^\prime$ cannot both be encoded with the same message from the sender.).
If $\chi(\C_s(G))$ is chromatic number of the confusion graph of $G$ and $\alpha(\C_s(G))$ is the size of the largest independent set in the confusion graph of $G$, then 
	$$ \b(G,s) = \log_s \chi(\C_s(G)) \qquad \mbox{and} \qquad \gn(G,s) = \log_s \alpha(\C_s(G)). $$ 
For any graph $H$, we have the identity $\chi(H)\alpha(H) \geq |H|$ and so we have the identity \cite{Riis2007graph}
	$$ \b(G,s) + \gn(G,s) \geq \log_s \left| \C_s(G) \right| = n. $$
We use this identity and the fact that the fractional-clique protocol $\P_{fcp}$ is optimal (Theorem~\ref{thm:n_large}) to prove Theorem~\ref{thm:defective}. This theorem in general is a new result, although the case $s=2$ was proven combinatorially in \cite{bar2011index}. Theorem~\ref{thm:defective} shows that the size of an optimal index code, $\b(G,s)$, depends on the factorisation structure of the size of the alphabet, $s$, used for the input. 

\begin{thm}
\label{thm:defective}
For a given $s$, let $b$ be the smallest factor of $s$ which is at least $\sqrt{s}$. There exists some $N$ such that for all odd $n>N$, 
	$$ \b(C_n,s) = \frac{n-1}{2} + \log_sb. $$
\end{thm}
\begin{proof}
Write $a=s/b$. First by Theorem~\ref{thm:n_large}, $\gn(C_n,s) = (n-1)/2 +\log_sa$ for all large enough odd $n$. Therefore, 
	$$ \b(C_n,s) \geq n - \gn(C_n,s) = \frac{n-1}{2} + \log_sb. $$
To show that we in fact get equality, we define a set of $bs^{(n-1)/2}$ possible messages with which the sender can solve the index coding with side information problem on $C_n$. Let $\phi$ and $\psi$ be defined as in Definition~\ref{defn:phi_and_psi}. This means that $\phi \times \psi$ is a bijection from $\Z_a \times \Z_b$ to $\Z_s$. Now for any colouring $c = (c_1,c_2, \ldots ,c_n) \in \Z_s^n$ let the sender broadcast the following values: 
\begin{itemize}
\item For $i=1,2,3, \ldots ,\tfrac{n-1}{2}$, the sender broadcasts the residue $\phi(c_{2i-1})+\phi(c_{2i})$ modulo $a$ and 
the residue $\psi(c_{2i})+\psi(c_{2i+1})$ modulo $b$.
\item Additionally, the sender broadcasts the residue $\psi(c_1)+\phi(c_n)$ modulo $b$.
\end{itemize}
The sender broadcasts $\frac{n-1}{2}$ residues modulo $a$ and $\frac{n+1}{2}$ residues modulo $b$, and so the total number of possible messages that the sender might send is 
	$$ m = a^{(n-1)/2}b^{(n+1)/2} = bs^{(n-1)/2}. $$
Furthermore, each receiver, $v_i$, knows $c_{i-1}$ and $c_{i+1}$, and so can recover both $c_i$ because she can recover both $\phi(c_i)$ and $\psi(c_i)$. 
\end{proof}

\section{Acknowledgements}
Puck Rombach is supported by AFOSR MURI Grant No. FA9550-10-1-0569 and ARO MURI Grant No. W911NF-11-1-0332.

\bibliography{gnCnARS}
\bibliographystyle{plain}

\end{document}